\documentclass[11pt]{article}
\usepackage[utf8]{inputenc}
\usepackage{amsmath, amssymb}
\usepackage{amsthm}
\usepackage[a4paper, margin=1in]{geometry}
\usepackage{setspace}
\usepackage{newpxtext}
\usepackage{authblk}
\usepackage[hyperindex,breaklinks]{hyperref}
\usepackage{xcolor}
\usepackage{tikz}
\usepackage{float}   
\usepackage{subcaption}
\usetikzlibrary{arrows.meta, positioning, shapes.geometric}

\newcommand{\Exp}[1]{\ensuremath{\mathsf{\mathsf{E}}\left[#1\right]}}

\DeclareMathOperator*{\argmax}{arg\,max}

\newtheorem{proposition}{Proposition}

\title{The unintended consequences of large language models as a labor-augmenting technology in science }

\author[1,2]{Eamon Duede\thanks{Email: eduede@purdue.edu}}
\author[3]{Kevin Gross\thanks{Email: kevin\_gross@ncsu.edu}}
\author[2]{M.J. Crockett\thanks{Email: mj.crockett@princeton.edu}}
\author[4,5]{Carl T. Bergstrom\thanks{Email: cbergst@uw.edu}}
\affil[1]{Purdue University}
\affil[2]{Princeton University}
\affil[3]{North Carolina State University}
\affil[4]{University of Washington}
\affil[5]{Santa Fe Institute}

\date{} % Optional

\begin{document}

\maketitle

\begin{abstract}
As a labor-augmenting technology, large language models (LLMs) have the potential to accelerate scientific activity across the research pipeline. But even if LLMs perform on par with human experts at selected tasks, their use will bring unintended consequences as they alter the balance of frictions and inducements that steer the allocation of research effort across projects. Here we develop a simple mathematical model to illustrate. In fields where LLMs are useful primarily as tools for discovering promising projects, researchers will become more selective about what they publish; where they facilitate the process of publishing existing data, researchers will become less selective. By allowing scientists to work more quickly, LLMs raise the opportunity cost of researcher time, creating incentives to refine papers less thoroughly before moving on. Enticing as it is to imagine that, by saving us time on mundane tasks, LLMs will provide us with more time to think deeply and develop projects completely, our results temper such hopes.
\end{abstract}

\section{Introduction}

Researchers are highly polarized about whether and how large language models (LLMs) can and should impact the practice of science \cite{binz2025should}. Some argue that LLMs already serve as valuable research tools. Unlike other novel technologies that accelerate only one aspect of a workflow, they claim, LLMs are poised to transform the entire research pipeline, from hypothesis generation \cite{si2025can,zhou2024hypothesis} and experiment design \cite{boiko2023autonomous,gu2024blade}, to processing, synthesizing and analyzing data \cite{majumder2025discoverybench}; from performing literature reviews \cite{wang2024autosurvey,asai2026synthesizing}, to writing up results \cite{liang2025quantifying,lu2024ai, siler2026diffusion}. Others point out that current LLMs have significant limitations that may be insurmountable: they have been shown to hallucinate plausible sounding but false content \cite{ji2023survey,huang2025survey,lin2022truthfulqa}, fabricate bibliographic citations to nonexistent scholarship \cite{walters2023fabrication,zhao2026llm}, overgeneralize scientific results \cite{peters2025generalization}, engage in sycophancy \cite{sharma2024towards}, drive deskilling \cite{budzyn2025endoscopist}, and misrepresent how conclusions are reached \cite{turpin2023language}. 

Critics have addressed the systemic harms that result from injudicious use of flawed LLMs. But even well-functioning LLMs can have unintended consequences, for example homogenizing inquiry \cite{messeri2024artificial}. As we demonstrate here, they also risk disrupting the balance of frictions and inducements that scientific norms and institutions impose \cite{mertonPrioritiesScientificDiscovery1957,kitcher1990,dasgupta1994toward,strevens_role_2003,zollman2009optimal,hull2010science} thereby reshaping scientific practice and the allocation of cognitive labor.  

To distinguish the effects of flawed technology from the consequences of labor-augmenting machines, we will model LLMs as imagined by their most enthusiastic proponents: capable of reducing time-costs without increasing errors and mistakes, at negligible financial cost.

Our model considers two key distinctions. First, \textit{when} do LLMs contribute? Do they accelerate the research process before the scientific value of a line of inquiry is known, or afterward? Second, do LLMs ameliorate the fixed costs of producing a manuscript (constructing figures, formatting bibliographies, drafting a cover letter) or do they allow more rapid development of a study beyond the bare minimum necessary for publication (e.g. generalizing theoretical results, conducting a sensitivity analysis, or excluding plausible counter-examples)? 

\section{Model}

Drawing upon classic results from optimal foraging theory in behavioral ecology \cite{charnov1976optimal}, we develop a simple model of how scientists allocate labor to various projects (Figure \ref{fig:equilibrium_combined}). For simplicity, we consider time as the only limiting resource; one could readily extend the model to consider production functions that incorporate other material inputs.

The project lifecycle is as follows: Initially, a researcher does not know how fruitful a line of inquiry will be. After a \textit{discovery phase} of fixed length $s>0$, they learn the value $v$ of their project. We envision this initial discovery phase as including activities such as identifying hypotheses, designing and conducting experiments, and evaluating preliminary data. 

Once a project's value is known, the researcher decides whether to abandon the investigation or to develop it further. If the researcher abandons the project, they receive no payoff and immediately begin the discovery phase of a new project. If they develop it, they invest further time $t=a+b$ in a \textit{development phase}. The time $a\geq0$ represents \textit{required work} that is necessary to publish a paper : creating figures, typesetting, proofreading, and submission chores, for example. The time $b\geq0$ represents additional \textit{discretionary development} during which the researchers further improve the project by conducting follow-up experiments, deepening statistical analysis, polishing prose, etc.

% \begin{figure}[H]
%         \centering
%         \includegraphics[width=.75 \linewidth]{figs/Model_diagram.pdf}
%         \caption{A schematic diagram of the model with a discovery phase of length $s$ followed by an optional development phase of length $t=a+b$.}
  
%     \label{fig:model}
% \end{figure}

Upon investing time $t$ in the development phase of a project, researchers obtain a benefit $v \, u(t)$.  The saturating function $u(t)$ captures how thoroughly the project is developed, and is quantified as the proportion of the maximum available scientific value $v$ that the researcher obtains. It has value zero until it reaches the minimum necessary investment $a$, and is increasing but concave thereafter.

Researchers select a policy that specifies how much time to invest in developing a value-$v$ project. The optimal policy maximizes the rate at which benefits accrue over the long haul of conducting many projects. We investigate changes in development time and the degree of development when LLMs speed up science, either by shortening $s$, shortening $a$, or accelerating the rate at which benefits accrue with increasing $b$.

\section{Analysis}

Our analysis centers on how accelerating distinct phases of the research pipeline affects which projects get developed and how thoroughly they are developed.  We study each phase separately, recognizing that LLM technology will likely affect several phases at once.   Proofs appear in the \hyperref[Appendix]{Appendix}. 

\medskip

\noindent {\bf Optimum policy}.   The researcher's long-run rate of return, which we denote by $\lambda$, sets the opportunity cost of their time. When deciding whether or not to develop a project, the researcher compares this opportunity cost against the gain from developing a project. A unique optimal policy exists in which (1) any project that the researcher chooses to publish will be developed up to the point that marginal value equals the long-run rate of return (Figure \ref{fig:equilibrium_combined}B), and (2) there is a threshold project value $v_0$ above which projects are developed and below which they abandoned (Figure \ref{fig:equilibrium_combined}C). Higher-valued projects are developed more extensively (Figure \ref{fig:equilibrium_combined}D). 
 
\begin{figure}
    \centering

    % Full-width figure at the top
    \begin{subfigure}[t]{\linewidth}
        \centering
        \includegraphics[
            width=\linewidth,
            height=0.14 \textheight,
            keepaspectratio
        ]{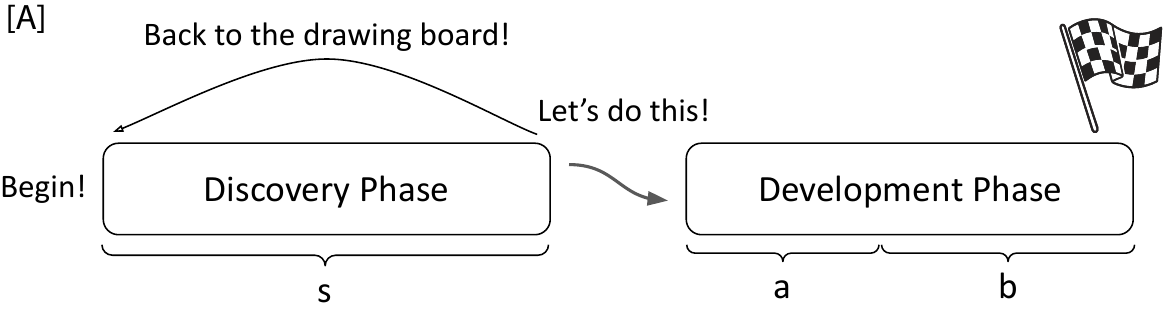}
        \label{fig:top_figure}
    \end{subfigure}
    \vspace{0.5em}

    \begin{subfigure}[t]{\linewidth}
        \centering
        \includegraphics[width=\linewidth]{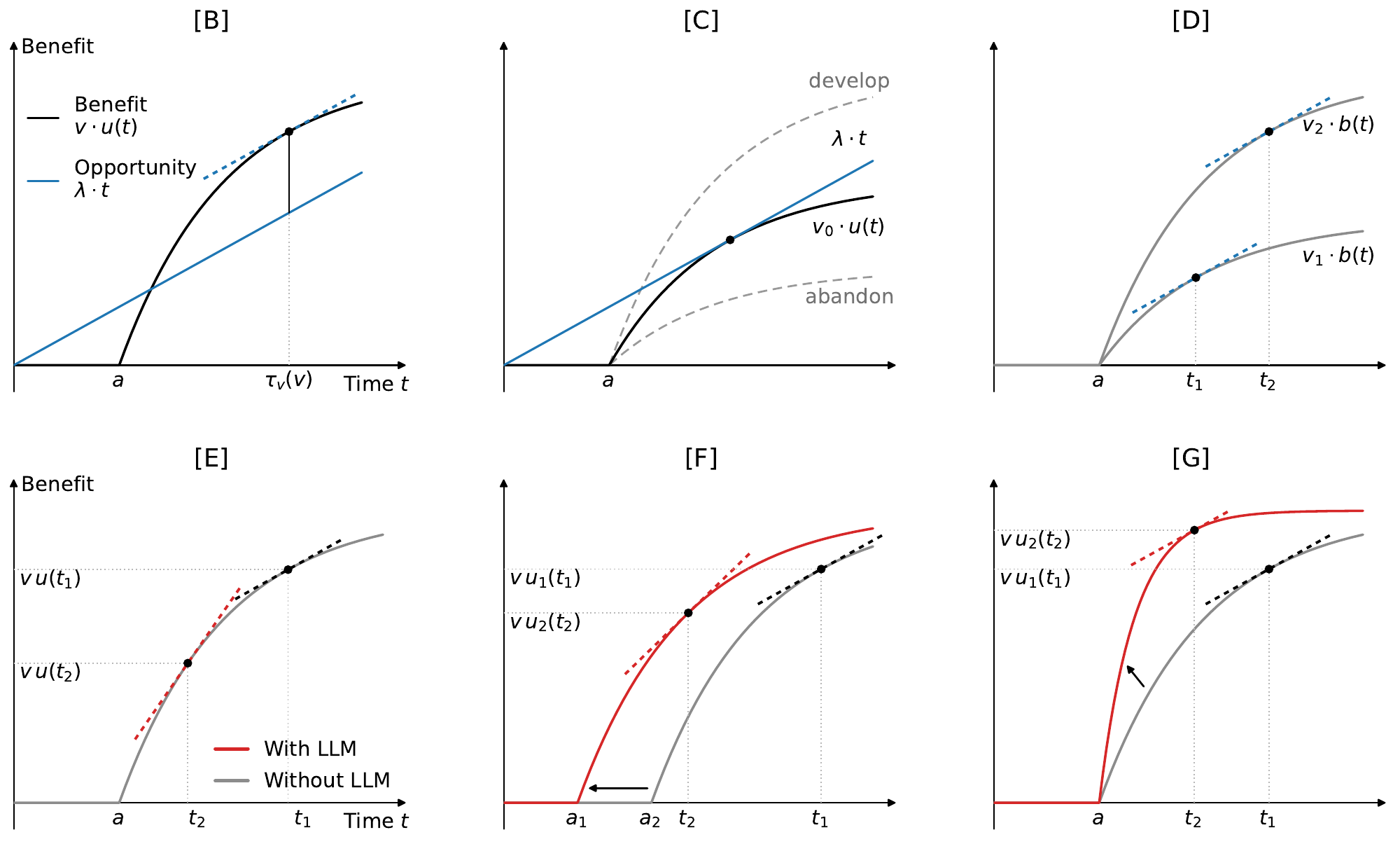}
        \label{fig:authordecision}
    \end{subfigure}
    \hfill

    \caption{Modeling optimal policies. [A] Schematic diagram of the model with a discovery phase of length $s$ followed by an optional development phase of length $t=a+b$. [B] The researcher chooses the amount of time to invest in a project of value $v$. At the optimal investment, the marginal benefit of effort equals the opportunity cost of the researcher's time.  [C] There exists a unique threshold value $v_0$ such that the line $\lambda\,t$ is tangent to the curve $v_0\,u(t)$. For $v$ values above $v_0$ authors will develop a project; below $v_0$ they will abandon it. [D] Higher valued projects are developed more thoroughly. 
    A reduction in $s$ [E] or $a$ [F] increases the opportunity cost of time so any developed project will be developed less thoroughly.~[G] When an LLM accelerates the discretionary development phase, any developed project will be developed more thoroughly.}
    \label{fig:equilibrium_combined}

\end{figure}

\medskip

\noindent{\bf Shortening the discovery period.} When LLMs shorten the discovery period $s$, the researcher accomplishes more in less time and thus the opportunity cost of time $\lambda$ increases. Researchers become selective about the projects they decide to develop, because as $s$ decreases, it is less costly to return to the drawing board. Not only do researchers now develop a smaller fraction of projects; the projects that they do develop are developed less thoroughly (Figure~\ref{fig:equilibrium_combined})E), due to the increased opportunity cost of time.
\medskip

\noindent {\bf Reduction in minimum development time}.
Reducing minimum development time $a$ likewise increases opportunity of cost of time $\lambda$.   However, it also reduces the minimum time required to develop a publishable paper.  These two effects push in opposite directions on $v_0$. At the margin of indifference, the latter effect always dominates the former, leading to a drop in $v_0$. Developed projects are again developed less thoroughly (Figure~\ref{fig:equilibrium_combined}F), because of the increased opportunity cost of time $\lambda$.

\medskip

\noindent {\bf Acceleration of discretionary effort}.
When an LLM accelerates the pace at which discretionary development improves a project, the opportunity cost of time increases, and the researcher's benefit curve is deformed upward and to the left (Figure~\ref{fig:equilibrium_combined}F).  As with a decrease in $a$, these effects push in opposite directions the threshold for developing a paper. Here the net effect on the threshold $v_0$ is ambiguous without further structure (see \hyperref[Appendix]{Appendix}). 
Accelerating the discretionary phase unambiguously increases the thoroughness to which projects are developed, because it increases the real value of discretionary development more than it increases the opportunity cost of the researcher's time (Figure~\ref{fig:equilibrium_combined}G).

\section{Discussion}

As labor-augmenting technology, LLMs shift the costs associated with producing research outputs. Our model suggests that, attuned to these costs, scientists will respond by reallocating their effort accordingly and depending on where in the research cycle LLMs have their biggest impact. When LLMs help researchers quickly identify the best projects, as in some technical fields, researchers will become more discriminating about which projects they develop. On the other hand, if LLMs mostly accelerate writing and analysis, as in fieldwork-based disciplines, the fraction of investigations that are published may increase.

The common force uniting these results is that in our model, LLMs make researchers' time more valuable. This is also why researchers become less willing to refine and further develop already-publishable results: there is something more useful they could be doing with that time. 

We have treated scientific institutions and incentives as static. This is reasonable on a short timescale; individual investigators change their actions faster than institutions change their policies. With the rapid development of LLM technology and widespread though heterogeneous adoption by researchers, a period of mismatch between institutions and practices is inevitable. 

Eventually, however, institutions will have to respond to LLM uptake \cite{kusumegi2025scientific,evans2025after}. For example, journal submissions are increasing rapidly in fields where LLMs reduce the time it takes to produce a manuscript \cite{messeri2026uncritical}, straining the peer review system \cite{bergstrom2026screening}.  Yet institutional responses must be sensitive to disciplinary differences. The effect of LLM technology on research output is not merely uniform acceleration, but instead will likely have nuanced affects based on which aspects of research are accelerated by LLMs. 

As researchers, we spend much of our time on tedious workaday tasks. It is enticing to suppose that by automating some of these, LLMs will enable us to reclaim our time and use the surplus to to think more deeply and develop our investigations more thoroughly. Unfortunately, our model reveals a flaw in this logic.  As a labor-augmenting technology, LLMs increase the opportunity cost of our time, impelling us to do more, less well---rather than the same amount, better.

\section*{Acknowledgments}
\thanks{
The idea for this paper emerged from the (Re)designing AI for Diverse Disciplines workshop at the Santa Fe Institute in October 2024. The authors thank the participants in that workshop, particularly Michael Strevens, for helpful comments and suggestions. Funding for this work was provided the Alfred P. Sloan Foundation (G-2024-22468 and G-2025-79234 to ED), the National Science Foundation (SES-2346645 to CTB and SES-2346644 to KG), and the Templeton World Charity Foundation (AWD-023376 and AWD-023376 to CB). }

\section*{AI use statement}
The text of this paper was written without AI assistance beyond the spell-checking and grammar-checking features built into Overleaf. The figures were drafted with pen and paper and hand-coded in Mathematica, but CTB used ChatGPT 5.5 assistance to refine fonts, label positions, etc., as well as for the help with the intricacies of \LaTeX ~formatting. Literature search was predominantly conducted using traditional tools such as Google Scholar, Google search, Semantic Scholar, and Web of Science but such tools now integrate AI. We wrote all proofs, but CTB and KG worked with ChatGPT 5.4, ChatGPT 5.5, and Gemini 3.1 to suggest proof strategies, and ED worked with Claude Fable 5 to check the proofs for mathematical and notational consistency and mistakes. MJC did not use LLMs in their contributions to the paper.
\section{Appendix}
\label{Appendix}
\subsection{Model}

A formal analysis of the model requires the following additional structure.  First, assume that $v$, the maximal value of a project, takes a distribution $F$ across projects.  Assume that $F$ is atomless and has support on the interval $[0, \bar{v}]$.  Let $V$ denote a random variable with distribution $F$.  

Assume also that the function $u(t)$ is continuously differentiable with the following properties: $u(t)=0$ for $t \leq a$ (there is no payoff to investing development time $\leq a$); $u'>0$ for $t>a$ (payoff strictly increases with discretionary development); $u''<0$ for $t>a$ (additional discretionary development yields diminishing returns); and $u(t)$ approaches 1 from below as $t$ gets large (maximum possible benefit is $v$).  

\subsection{Base case}

A policy $\tau(v)$ is a function that gives the time invested in developing a value-$v$ project, where $\tau(v)=0$ indicates that the project is abandoned and $\tau(v)>0$ indicates that the project is developed.  The renewal-reward theorem \cite{ross2014introduction} states that the payoff the policy $\tau$ will equal the average benefit per project divided by the average time spent on each project, i.e., 
\begin{equation}
\dfrac{\Exp{V \, u(\tau(V))}}{s + \Exp{\tau(V)}}.
\label{eq:lambda-definition}
\end{equation}
Notice that this holds whether the research carries out the projects sequentially or ``multitasks'', switching effort among multiple ongoing projects. Write the payoff for any policy $\tau$ as the function $\Lambda(\tau)$.  The researcher's optimal policy $\tau^*$ is the policy that maximizes this payoff function.  The following proposition characterizes the optimal policy.

\begin{proposition} 
There exists a payoff-optimal policy $\tau^*(v)$ that is unique except for a single $v$ at which the researcher is indifferent between developing the project or not.  There is an associated optimal payoff $\lambda^*$.  Moreover, the optimal policy is characterized by 
\[
\tau^*(v) \in \argmax_{t \geq 0} \left\{vu(t) - \lambda^* t \right\}.
\] 
Finally, the optimal payoff is the unique value at which the cost of discovery exactly balances the expected net benefit from developing the project, i.e.,
\[
\lambda^* s = \Exp{Vu(\tau^*(V)) - \lambda^*\tau^*(V)}.
\]
\end{proposition}

\begin{proof}
Suppose that the shadow (opportunity) cost of time is $\lambda>0$, and write 
\[
m_{\lambda}(v) = \max_{t \geq 0} \left[v \,u(t)-\lambda \, t \right]
\]
as the net surplus for a value-$v$ project under $\lambda$.  Define 
\[
\Phi(\lambda) = \Exp{m_\lambda(V)} - \lambda s
\]
as the difference between the expected surplus and the cost of the discovery phase under $\lambda$.  

We first show that there is a unique value $\lambda^*$ such that $\Phi(\lambda^*) = 0$.  First, we establish that $\Phi(\lambda)$ is continuous in $\lambda$.  To do so, note that by our assumptions on $u(t)$, that $t=\argmax_{t \geq 0} \left[v\,u(t)-\lambda t\right]$ will solve the first-order condition $u'(t)=\lambda/v$.  Then we can invert $u'$ and plug into $m_\lambda$ to obtain $m_\lambda(v)=v\,u(u'^{-1}(\lambda/v)) -\lambda t$.  By continuity of $u$ and $u'$, $m_\lambda(v)$ is continuous in $\lambda$.  Next, because $|m_\lambda(v)|\leq\bar{v}$, then $\Exp{m_\lambda(V)}$ and hence $\Phi(\lambda)$ are also continuous in $\lambda$.  

Next, $\lim_{\lambda \downarrow 0} \Phi(\lambda) = \Exp{V}>0$, and $\lim_{\lambda \uparrow \infty} \Phi(\lambda) = -\infty$.  Finally, $m_\lambda(v)$ is weakly decreasing in $\lambda$ for all $v$, and thus $\Phi(\lambda)$ is strictly decreasing in $\lambda$.  By continuity and strict monotonicity, there is a unique $\lambda^*$ such that $\Phi(\lambda^*)=0$.

Next we show that $\lambda^*$ is at least as large as the payoff to any policy $\tau$.  First consider any policy $\tau$ and any $v$.  Because $m_{\lambda^*}(v)$ is the maximum surplus available under $\lambda^*$, then
\[
v u(\tau(v)) - \lambda^* \tau(v) \leq m_{\lambda^*}(v).
\]
Taking expectations gives
\[
\Exp{V u(\tau(V))} - \lambda^* \Exp{\tau(V)} \leq \Exp{m_{\lambda^*}(V)} = \lambda^* s.
\]
Re-arrange to give 
\[
\Exp{V u(\tau(V))} \leq \lambda^* (s + \Exp{\tau(V)}).
\]
Rearrange once more to find
\[
\Lambda(\tau) = \dfrac{\Exp{V u(\tau(V))}}{s + \Exp{\tau(V)}} \leq \lambda^*.
\]

Next we show that there is a policy $\tau^*$ that obtains $\lambda^*$.  Set $\tau^*(v) \in \argmax_{t\geq0} \left[v \,u(t)-\lambda^* \, t \right]$. Then
\[
v u(\tau^*(v)) - \lambda^* \tau^*(v) = m_{\lambda^*}(v).
\]
Taking expectations gives
\[
\Exp{V u(\tau^*(V))} - \lambda^* \Exp{\tau^*(V)} = \Exp{m_{\lambda^*}(v)} = \lambda^* s.
\]
Rearrange as before to give $\Lambda(\tau^*) = \lambda^*$.  Thus $\tau^*$ is optimal.

Uniqueness follows from observing that, by virtue of our assumptions on $u(t)$, $\tau^*(v)$ is only non-unique at $v_0$, the project at which a researcher is indifferent between developing and abandoning.  But this is a set of $F$-measure 0.
\end{proof}

In the remainder, we let $\tau$ denote the optimal policy and let $\lambda$ denote the optimal payoff.

\begin{proposition} The researcher's payoff $\lambda$ is strictly positive. \end{proposition}
\begin{proof}
    Consider a candidate policy in which every project is developed to the same discretionary development time $b$ for some $b>0$.  This policy yields a strictly positive payoff. But $\lambda$ is at least as large as this payoff, so $\lambda > 0$.
\end{proof}

\begin{proposition} The threshold project $v_0$ exists, is unique, and is interior.  \end{proposition}
\begin{proof}
    Write $\sigma_\lambda(v) = \max_{t \geq a} \left[v \,u(t)-\lambda \, t \right]$ as the net surplus to developing a value-$v$ project.  (Note that $\sigma_\lambda(v)$ differs from $m_\lambda(v)$, because $\sigma_\lambda$ is the maximum surplus to developing it, while $m_\lambda$ is the maximum surplus to developing a project or abandoning it.  Thus $m_\lambda(v) = \max(0, \sigma_\lambda(v))$.)  Of course $\sigma(0)=-\lambda a<0$.  There must be some projects that are worth developing,  otherwise we would have $\lambda=0$.  If some projects are worth developing, then the maximum-$v$ project must be worth developing, or $\sigma(\bar{v})>0$.  Thus continuity and strict monotonicity of $\sigma(v)$ imply the existence and uniqueness of $v_0$.
\end{proof}

\begin{proposition}
For developed projects ($v>v_0$), the optimal development time $\tau(v)$ and the thoroughness of development $u(\tau(v))$ strictly increase in $v$.
\end{proposition}

\begin{proof}
For developed projects, the optimal development time satisfies the first-order condition $v \, u'(\tau(v)) = \lambda$.  Strict concavity of $u$ for $t > a$ establishes that  $\tau(v)$  strictly increases in $v$.  Strict monotonicity of $u$ for $t>a$ then establishes that $u(\tau(v))$ strictly increase in $v$.
\end{proof}

\subsection{Comparative statics}

For these results, it is convenient to reparameterize the model with the function $z(b) = u(a+b)$ for $b \geq 0$, that is, $z(b)$ captures the portion of $u(t)$ that controls how the payoff increases with discretionary development time $b$. Now the researcher's post-$v$ problem is to choose $b\geq0$ to maximize $\left\{0, vz(b)-\lambda(a+b)\right\}$.  Write the surplus to developing a value-$v$ project at the optimum policy as $\sigma(v)  = \max_{b \geq 0} \left[v \,z(b)-\lambda \, (a+b) \right]$.  The researcher's policy $\beta(v)$ now gives the discretionary development time at $v$, where $\beta(v)=0$ is taken to mean that the project is abandoned.  (Note that all developed projects will have $\beta(v)>0$.) Finally, write the expected benefit of a project as $\mu = \Exp{V \, z(\beta(V))}$ and the expected time investment as $\rho = s + \Exp{\mathsf{1}_{V\geq v_0} \, (a+\beta(V))}$, so that the researcher's payoff is $\lambda = \mu/\rho$.

Occasionally we will also use $\beta_0$ to denote the discretionary development time of the threshold project $v_0$.

\subsubsection{Shortening the discovery phase}
 
\begin{proposition} $\lambda$ strictly decreases with $s$. 
\label{prop:l-alpha}
\end{proposition}
\begin{proof}
Suppose $s_1 < s_2$, let $\beta_i$ be the optimal policy under $s_i$, let $\Lambda_i(s)$ be the payoff to policy $\beta_i$ under $s$, and let $\lambda(s)$ be the maximal payoff under $s$.  Then $\lambda(s_1) = \Lambda_1(s_1) \geq \Lambda_2(s_1)>\Lambda_2(s_2)=\lambda(s_2)$, where the first inequality follows from the fact that $b_1$ is payoff-maximizing at $s_1$, and the second inequality follows from the fact that under any policy the researcher's payoff strictly decreases with $s$. 
\end{proof}

\begin{proposition} Decreasing $s$ increases $v_0$. \end{proposition}

\begin{proof}

    Write the available surplus at $v$ for a given $s$ as
    
    \begin{displaymath}
    \sigma(v,s) = \max_{b \geq 0} \left[v \,z(b)-\lambda \, (a+b) \right]
    \end{displaymath}
    Let $v_0(s)$ be the threshold project under $s$, such that $\sigma(v_0(s),s)=0$.  We analyze for (the sign of)
    \begin{displaymath}
        \dfrac{\partial \sigma(v_0(s),s)}{\partial s} 
    \end{displaymath}
    which gives the rate of change of the surplus as $s$ increases but while $v$ stays fixed at $v_0(s)$. By the envelope theorem, 
    \begin{displaymath}
        \dfrac{\partial \sigma(v_0(s),s)}{\partial s} = -\dfrac{d\lambda}{ds}(a+\beta_0)>0.
    \end{displaymath}
    Therefore as $s$ increases the researcher now develops the project at which they had previously been indifferent and thus the threshold project for indifference must drop.
\end{proof}

\begin{proposition} For any $v>v_0$, decreasing $s$ decreases $z(\beta(v))$. \end{proposition}
\begin{proof}
    When $v>v_0$, $\beta(v)$ satisfies the first-order condition $v \, z'(\beta(v)) = \lambda$.  By the strict concavity of $z(b)$ for $b>0$, an increase in $\lambda$ decreases $\beta(v)$ for any given $v$, and thus decreases $z(\beta(v))$ as well.
\end{proof}

\subsubsection{Shortening the minimum development time}

\begin{proposition} Decreasing $a$ increases $\lambda$. \end{proposition}
\begin{proof}
One can prove this proposition by applying an argument identical to the proof of Proposition \ref{prop:l-alpha}.  Simply replace $s$ with $a$ in the proof of Proposition \ref{prop:l-alpha}. But because we will need it for the next proposition, here we explicitly derive $d\lambda/da$.

Writing $\lambda = \mu/\rho$, we apply the envelope theorem to obtain
    \begin{equation}
        \dfrac{d\lambda}{da} =  \dfrac{-\mu\left(\frac{\partial\rho}{\partial a}\right)}{\rho^2} = -\dfrac{\lambda}{\rho}\dfrac{\partial \rho}{\partial a}.
    \end{equation}
    Now $\partial \rho/\partial a$ is the first-order effect of an increase in $a$ on the average project length, which equals the fraction of projects that are developed:
    \begin{displaymath}
      \dfrac{\partial \rho}{\partial a}  = \dfrac{\partial }{\partial a}\left[ s + \int_{v_0}^{\bar{v}} (a + \beta(v)) \, dF(v)\right] = 1 - F(v_0).
    \end{displaymath}
    Thus
     \begin{equation}
        \dfrac{d\lambda}{da} =  -\dfrac{\lambda}{\rho}[1-F(v_0)].
    \end{equation}
    Here $\lambda$, $\rho$, and $1-F(V_0)$ are all strictly positive and so the derivative is negative.  However, we can go further.  Observe that $\rho / (1 - F(v_0)) = s / (1 - F(v_0)) + \Exp{a+\beta(v)|v>v_0}$ gives the average time in between the completion of developed projects.  Write this time as $\rho / (1 - F(v_0))=T_d$ and substitute into the expression above to find
    \begin{equation}
        \dfrac{d\lambda}{da} =  -\dfrac{\lambda}{T_d} < 0.
    \end{equation}
\end{proof}

\begin{proposition} Increasing $a$ increases $v_0$. \end{proposition}
\begin{proof}

    Write the available surplus at $v$ for a given $a$ as
    \begin{displaymath}
    \sigma(v,a) = \max_{b \geq 0} \left[v \,z(b)-\lambda \, (a+b) \right]
    \end{displaymath}
    Let $v_0(a)$ be the threshold project under $a$, such that $\sigma(v_0(a),a)=0$. 
    We analyze for the sign of
    \begin{displaymath}
        \dfrac{\partial \sigma(v_0(a),a)}{\partial a} 
    \end{displaymath}
    which gives the rate of change of the surplus as $a$ increases but while $v$ stays fixed at $v_0(a)$.

   Apply the envelope theorem to find
    \begin{displaymath}
        \dfrac{\partial\sigma(v_0(a),a)}{\partial a} = -\lambda -\dfrac{d\lambda}{da}(a + \beta_0).
    \end{displaymath}
    Plug in ${d\lambda}/{da} = - \lambda / T_d$ from Proposition 8 into the above to yield 
    \begin{align*}
        \dfrac{\partial\sigma(v_0(a),a)}{\partial a} & =  -\lambda + \dfrac{\lambda(a + \beta_0)}{T_d} = -\lambda\left( 1 - \dfrac{a + \beta_0}{T_d}\right).
    \end{align*}
    But $a+\beta_0<T_d$, because $a+\beta_0$ is the development time given to the minimally developed project.  Thus $\partial \sigma(v_0(a),a) /\partial a<0$. Therefore as $a$ increases the researcher now abandons the project at which they had previously been indifferent and thus the threshold project for indifference must rise.
\end{proof}

\begin{proposition} For any $v>v_0$, decreasing $a$ decreases $z(\beta(v))$. \end{proposition}
\begin{proof}
    The first-order condition for $\beta(v)$ for $v>v_0$ is $v \, z'(\beta(v)) = \lambda$.  Thus $\beta(v)$ and also $z(\beta(v))$ must decrease as $\lambda$ increases.
\end{proof}

\subsubsection{Accelerating discretionary development}

\begin{proposition} In the acceleration model, if the elasticity of $z(b)$ is decreasing in $b$, then increasing $\alpha$ decreases $v_0$. \end{proposition}
\begin{proof}
    In the acceleration model, a researcher who has learned $v$ faces the problem of choosing $b$ to maximize $\left\{0, v z(\alpha b)-\lambda(a+b)\right\}$. Write the available surplus at $v$ for a given $\alpha$ as
    \begin{displaymath}
    \sigma(v,\alpha) = \max_{b \geq 0} \left[v \,z(\alpha b)-\lambda \, (a+b) \right]
    \end{displaymath}
    Let $v_0(\alpha)$ be the threshold project under $\alpha$, such that $\sigma(v_0(\alpha),\alpha)=0$. Our aim will be to determine the sign of $\partial \sigma(v_0(\alpha),\alpha )/ {\partial \alpha} $
    which gives the rate of change of the surplus as $\alpha$ increases but while $v$ stays fixed at $v_0(\alpha)$.

    By the envelope theorem,
    \begin{displaymath}
        \dfrac{\partial \sigma(v_0(\alpha),\alpha)}{\partial \alpha} = v_0 \beta_0 z'(\alpha \beta_0)-\dfrac{d\lambda}{d\alpha}(a+\beta_0).
    \end{displaymath}

    To find $d\lambda/d\alpha$, we appeal to the envelope theorem a second time. Write $\mu = \Exp{V \, z(\alpha \beta(V))}$ as the researcher's average benefit per project, and continue to write $\rho$ as the average time investment per project, such that $\lambda = \mu/\rho$. Apply the envelope theorem to obtain
    \begin{equation*}
        \dfrac{d\lambda}{d\alpha} =  \dfrac{\Exp{V \, \beta(V) \, z'(\alpha \beta(V))}}{\rho}.
    \end{equation*}
    Plug in to yield 
    \begin{displaymath}
        \dfrac{\partial \sigma(v_0(\alpha),\alpha)}{\partial \alpha} = v_0 \beta_0 z'(\alpha \beta_0)- \Exp{V \, \beta(v) \, z'(\alpha \beta(V))} \dfrac{a+\beta_0}{\rho}.
    \end{displaymath}
    To rewrite the final term of this expression we use the definition of the long-run payoff and fact that at $v_0$ the surplus is 0:
    \begin{displaymath}
         \dfrac{\Exp{V z(\alpha \beta(V))}}{\rho}=\lambda = \dfrac{v_0 z(\alpha \beta_0)}{a+\beta_0}.
    \end{displaymath}
    Thus $(a+\beta_0)/\rho = v_0 z(\alpha \beta_0)/\Exp{V z(\alpha \beta(V))}$, and we can write
    \begin{displaymath}
        \dfrac{\partial \sigma(v_0(\alpha),\alpha)}{\partial \alpha} = v_0 \beta_0 z'(\alpha \beta_0)- \Exp{V \, \beta(V) \, z'(\alpha \beta(V))} \dfrac{v_0 z(\alpha \beta_0)}{\Exp{V z(\alpha \beta(V))}}.
    \end{displaymath}
    Let $\varepsilon_z$ be the elasticity of $z$ and for $b>0$ write $q_\alpha$ as
    \begin{displaymath}
         q_\alpha(b)
      =
      \frac{b z'(\alpha b)}{z(\alpha b)}
      =
     \frac{1}{\alpha}\varepsilon_z(\alpha b).
    \end{displaymath}
    This gives us 
    \begin{displaymath}
\frac{\partial \sigma(v_0(\alpha),\alpha)}{\partial \alpha}
=
v_0 z(\alpha \beta_0)
\left[
q_\alpha(\beta_0)
-
\frac{
\mathbb{E}\!\left[
V z(\alpha \beta(V)) q_\alpha(\beta(V))
\right]
}{
\mathbb{E}\!\left[
V z(\alpha \beta(V))
\right]
}
\right],
\end{displaymath}
The term in square brackets compares the elasticity at the threshold project with a benefit-weighted average of elasticities across developed projects. Because the elasticity $\varepsilon_z$ is decreasing and $\beta(V)>\beta_0$ for every $V>v_0$,
$q_\alpha(\beta_0)$ exceeds this benefit-weighted average. Therefore 
\begin{displaymath}
\frac{\partial \sigma(v_0(\alpha),\alpha)}{\partial \alpha}>0.
\end{displaymath}

 As $\alpha$ increases, the previous threshold project $v_0$ returns a positive payoff and the new threshold project value $v_0$ must decrease.

    \end{proof}
    
\begin{proposition} For any $v>v_0$, increasing $\alpha$ in the acceleration model increases $z(\alpha \beta(v))$. \end{proposition}
\begin{proof}
    For any $v>v_0$ in the acceleration model, the first-order condition for each project's optimal development time is $v \alpha z'(\alpha \beta(v)) = \lambda$, which rearranges to $v z'(\alpha \beta(v)) = \lambda/\alpha$.  Thus, $z'(\alpha \beta(v))$ increases with $\alpha$ iff $\lambda/\alpha$ decreases with $\alpha$.  By simple calculus,
    \begin{displaymath}
        \dfrac{d}{d\alpha}  \left( \dfrac{\lambda}{\alpha} \right)= \dfrac{1}{\alpha}\left(\dfrac{d\lambda}{d\alpha} - \dfrac{\lambda}{\alpha}\right).
    \end{displaymath}
    By using our previous result for $d\lambda/d\alpha$ and applying the first-order condition, we have
    \begin{displaymath}
        \dfrac{d\lambda}{d\alpha} = \dfrac{\Exp{V \, \beta(V) \, z'(\alpha \beta(V))}}{\rho} = \dfrac{\lambda \Exp{\beta(V)}}{\alpha \rho}
    \end{displaymath}
    Plug in  to obtain
    \begin{displaymath}
        \dfrac{d}{d\alpha}  \left( \dfrac{\lambda}{\alpha} \right)= \dfrac{\lambda}{\alpha^2}\left(\dfrac{\Exp{\beta(V)}}{\rho} - 1\right).
    \end{displaymath}
    But $\Exp{\beta(V)}$ (the average discretionary time per project, developed or not) is certainly less than $\rho$ (the average total time per project), and thus $d(\lambda/\alpha)/d\alpha < 0$.
\end{proof}

\bibliographystyle{alpha}
\bibliography{bibliography}

@article{mertonPrioritiesScientificDiscovery1957,
    title = {Priorities in {Scientific} {Discovery}: {A} {Chapter} in the {Sociology} of {Science}},
    volume = {22},
    issn = {00031224},
    shorttitle = {Priorities in {Scientific} {Discovery}},
    url = {http://www.jstor.org/stable/2089193?origin=crossref},
    doi = {10.2307/2089193},
    number = {6},
    urldate = {2026-01-12},
    journal = {American Sociological Review},
    author = {Merton, Robert K.},
    month = dec,
    year = {1957},
    pages = {635},
}

@article{strevens_role_2003,
    title = {The {Role} of the {Priority} {Rule} in {Science}},
    volume = {100},
    url = {https://www.pdcnet.org/pdc/bvdb.nsf/purchase?openform&fp=jphil&id=jphil_2003_0100_0002_0055_0079},
    doi = {10.5840/jphil2003100224},
    language = {en},
    number = {2},
    urldate = {2026-01-12},
    journal = {The Journal of Philosophy},
    author = {Strevens, Michael},
    month = feb,
    year = {2003},
    pages = {55--79},
}

@article{zollman2009optimal,
  title={Optimal publishing strategies},
  author={Zollman, Kevin JS},
  journal={Episteme},
  volume={6},
  number={2},
  pages={185--199},
  year={2009},
  publisher={Cambridge University Press}
}

@article{dasgupta1994toward,
  title={Toward a new economics of science},
  author={Dasgupta, Partha and David, Paul A},
  journal={Research policy},
  volume={23},
  number={5},
  pages={487--521},
  year={1994},
  publisher={Elsevier}
}

@article{kitcher1990,
  title={The division of cognitive labor},
  author={Kitcher, Philip},
  journal={The Journal of Philosophy},
  volume={87},
  number={1},
  pages={5--22},
  year={1990}
}

@BOOK{hull2010science,
  title = {Science as a process: an evolutionary account of the social and conceptual
	development of science},
  publisher = {University of Chicago Press},
  year = {2010},
  author = {Hull, David L}
}

@article{messeri2024artificial,
  title={Artificial intelligence and illusions of understanding in scientific research},
  author={Messeri, Lisa and Crockett, MJ},
  journal={Nature},
  volume={627},
  number={8002},
  pages={49--58},
  year={2024},
  publisher={Nature Publishing Group UK London}
}

@article{binz2025should,
  title={How should the advancement of large language models affect the practice of science?},
  author={Binz, Marcel and Alaniz, Stephan and Roskies, Adina and Aczel, Balazs and Bergstrom, Carl T and Allen, Colin and Schad, Daniel and Wulff, Dirk and West, Jevin D and Zhang, Qiong and others},
  journal={Proceedings of the National Academy of Sciences},
  volume={122},
  number={5},
  pages={e2401227121},
  year={2025},
  publisher={National Academy of Sciences}
}

@article{ji2023survey,
  title={Survey of hallucination in natural language generation},
  author={Ji, Ziwei and Lee, Nayeon and Frieske, Rita and Yu, Tiezheng and Su, Dan and Xu, Yan and Ishii, Etsuko and Bang, Ye Jin and Madotto, Andrea and Fung, Pascale},
  journal={ACM computing surveys},
  volume={55},
  number={12},
  pages={1--38},
  year={2023},
  publisher={ACM New York, NY}
}

@inproceedings{si2025can,
  title={Can llms generate novel research ideas? a large-scale human study with 100+ nlp researchers},
  author={Si, Chenglei and Yang, Diyi and Hashimoto, Tatsunori},
  booktitle={International Conference on Learning Representations},
  volume={2025},
  pages={94003--94092},
  year={2025}
}

@inproceedings{zhou2024hypothesis,
  title={Hypothesis generation with large language models},
  author={Zhou, Yangqiaoyu and Liu, Haokun and Srivastava, Tejes and Mei, Hongyuan and Tan, Chenhao},
  booktitle={Proceedings of the 1st Workshop on NLP for Science (NLP4Science)},
  pages={117--139},
  year={2024}
}

@article{boiko2023autonomous,
  title={Autonomous chemical research with large language models},
  author={Boiko, Daniil A and MacKnight, Robert and Kline, Ben and Gomes, Gabe},
  journal={Nature},
  volume={624},
  number={7992},
  pages={570--578},
  year={2023},
  publisher={Nature Publishing Group UK London}
}

@inproceedings{majumder2025discoverybench,
  title={Discoverybench: Towards data-driven discovery with large language models},
  author={Majumder, Bodhisattwa Prasad and Surana, Harshit and Agarwal, Dhruv and Dalvi Mishra, Bhavana and Meena, Abhijeetsingh and Prakhar, Aryan and Vora, Tirth and Khot, Tushar and Sabharwal, Ashish and Clark, Peter},
  booktitle={International Conference on Learning Representations},
  volume={2025},
  pages={4556--4579},
  year={2025}
}

@article{wang2024autosurvey,
  title={Autosurvey: Large language models can automatically write surveys},
  author={Wang, Yidong and Guo, Qi and Yao, Wenjin and Zhang, Hongbo and Zhang, Xin and Wu, Zhen and Zhang, Meishan and Dai, Xinyu and Zhang, Min and Wen, Qingsong and others},
  journal={Advances in neural information processing systems},
  volume={37},
  pages={115119--115145},
  year={2024}
}

@article{asai2026synthesizing,
  title={Synthesizing scientific literature with retrieval-augmented language models},
  author={Asai, Akari and He, Jacqueline and Shao, Rulin and Shi, Weijia and Singh, Amanpreet and Chang, Joseph Chee and Lo, Kyle and Soldaini, Luca and Feldman, Sergey and D’Arcy, Mike and others},
  journal={Nature},
  pages={1--7},
  year={2026},
  publisher={Nature Publishing Group UK London}
}

@article{liang2025quantifying,
  title={Quantifying large language model usage in scientific papers},
  author={Liang, Weixin and Zhang, Yaohui and Wu, Zhengxuan and Lepp, Haley and Ji, Wenlong and Zhao, Xuandong and Cao, Hancheng and Liu, Sheng and He, Siyu and Huang, Zhi and others},
  journal={Nature Human Behaviour},
  pages={1--11},
  year={2025},
  publisher={Nature Publishing Group UK London}
}

@article{lu2024ai,
  title={The ai scientist: Towards fully automated open-ended scientific discovery},
  author={Lu, Chris and Lu, Cong and Lange, Robert Tjarko and Foerster, Jakob and Clune, Jeff and Ha, David},
  journal={arXiv preprint arXiv:2408.06292},
  year={2024}
}

@inproceedings{gu2024blade,
  title={Blade: Benchmarking language model agents for data-driven science},
  author={Gu, Ken and Shang, Ruoxi and Jiang, Ruien and Kuang, Keying and Lin, Richard-John and Lyu, Donghe and Mao, Yue and Pan, Youran and Wu, Teng and Yu, Jiaqian and others},
  booktitle={Findings of the Association for Computational Linguistics: EMNLP 2024},
  pages={13936--13971},
  year={2024}
}

@article{huang2025survey,
  title={A survey on hallucination in large language models: Principles, taxonomy, challenges, and open questions},
  author={Huang, Lei and Yu, Weijiang and Ma, Weitao and Zhong, Weihong and Feng, Zhangyin and Wang, Haotian and Chen, Qianglong and Peng, Weihua and Feng, Xiaocheng and Qin, Bing and others},
  journal={ACM Transactions on Information Systems},
  volume={43},
  number={2},
  pages={1--55},
  year={2025},
  publisher={ACM New York, NY}
}

@inproceedings{lin2022truthfulqa,
  title={Truthfulqa: Measuring how models mimic human falsehoods},
  author={Lin, Stephanie and Hilton, Jacob and Evans, Owain},
  booktitle={Proceedings of the 60th annual meeting of the association for computational linguistics (volume 1: long papers)},
  pages={3214--3252},
  year={2022}
}

@article{walters2023fabrication,
  title={Fabrication and errors in the bibliographic citations generated by ChatGPT},
  author={Walters, William H and Wilder, Esther Isabelle},
  journal={Scientific Reports},
  volume={13},
  number={1},
  pages={14045},
  year={2023},
  publisher={Nature Publishing Group UK London}
}

@inproceedings{sharma2024towards,
  title={Towards understanding sycophancy in language models},
  author={Sharma, Mrinank and Tong, Meg and Korbak, Tomek and Duvenaud, David and Askell, Amanda and Bowman, Sam and Durmus, Esin and Hatfield-Dodds, Zac and Johnston, Scott and Kravec, Shauna and others},
  booktitle={International Conference on Learning Representations},
  volume={2024},
  pages={110--144},
  year={2024}
}

@article{turpin2023language,
  title={Language models don't always say what they think: Unfaithful explanations in chain-of-thought prompting},
  author={Turpin, Miles and Michael, Julian and Perez, Ethan and Bowman, Samuel},
  journal={Advances in Neural Information Processing Systems},
  volume={36},
  pages={74952--74965},
  year={2023}
}

@article{siler2026diffusion,
  title={The diffusion of large language models in published academic articles},
  author={Siler, Kyle},
  journal={Proceedings of the National Academy of Sciences},
  volume={123},
  number={22},
  pages={e2605754123},
  year={2026},
  publisher={National Academy of Sciences}
}

@article{zhao2026llm,
  title={LLM hallucinations in the wild: Large-scale evidence from non-existent citations},
  author={Zhao, Zhenyue and Wang, Yihe and Stuart, Toby and De Vaan, Mathijs and Ginsparg, Paul and Yin, Yian},
  journal={arXiv preprint arXiv:2605.07723},
  year={2026}
}

@book{ross2014introduction,
  author    = {Ross, Sheldon M.},
  title     = {Introduction to Probability Models},
  edition   = {11},
  publisher = {Academic Press},
  address   = {Amsterdam},
  year      = {2014},
  isbn      = {9780124081215},
  note      = {Chapter 7: Renewal Theory and Its Applications}
}

@article{peters2025generalization,
  author  = {Peters, Uwe and Chin-Yee, Benjamin},
  title   = {Generalization bias in large language model summarization of scientific research},
  journal = {Royal Society Open Science},
  year    = {2025},
  volume  = {12},
  number  = {4},
  pages   = {241776},
  doi     = {10.1098/rsos.241776}
}

@article{charnov1976optimal,
  title={Optimal foraging, the marginal value theorem},
  author={Charnov, Eric L},
  journal={Theoretical population biology},
  volume={9},
  number={2},
  pages={129--136},
  year={1976},
  publisher={Elsevier}
}

@article{budzyn2025endoscopist,
  title={Endoscopist deskilling risk after exposure to artificial intelligence in colonoscopy: a multicentre, observational study},
  author={Budzy{\'n}, Krzysztof and Roma{\'n}czyk, Marcin and Kitala, Diana and Ko{\l}odziej, Pawe{\l} and Bugajski, Marek and Adami, Hans O and Blom, Johannes and Buszkiewicz, Marek and Halvorsen, Natalie and Hassan, Cesare and others},
  journal={The Lancet Gastroenterology \& Hepatology},
  volume={10},
  number={10},
  pages={896--903},
  year={2025},
  publisher={Elsevier}
}

@article{bergstrom2026screening,
  title={Screening, sorting, and the feedback cycles that imperil peer review},
  author={Bergstrom, Carl T and Gross, Kevin},
  journal={PLoS Biology},
  volume={24},
  number={2},
  pages={e3003650},
  year={2026},
  publisher={Public Library of Science San Francisco, CA USA}
}

@article{kusumegi2025scientific,
  title={Scientific production in the era of large language models},
  author={Kusumegi, Keigo and Yang, Xinyu and Ginsparg, Paul and de Vaan, Mathijs and Stuart, Toby and Yin, Yian},
  journal={Science},
  volume={390},
  number={6779},
  pages={1240--1243},
  year={2025},
  publisher={American Association for the Advancement of Science}
}

@article{evans2025after,
  title={After science},
  author={Evans, James and Duede, Eamon},
  journal={Science},
  volume={390},
  number={6774},
  pages={eaec7650},
  year={2025},
  publisher={American Association for the Advancement of Science}
}

@article{messeri2026uncritical,
  title={The uncritical adoption of AI in science is alarming—we urgently need guard rails},
  author={Messeri, Lisa and Crockett, MJ},
  journal={Nature},
  volume={653},
  number={8115},
  pages={675--676},
  year={2026},
  publisher={Nature Publishing Group UK London}
}

\end{document}